\documentclass{llncs}
\usepackage{amsmath}
\usepackage{amssymb}
\usepackage{mathrsfs}
\usepackage{graphicx}
\usepackage{color}
\usepackage[all]{xy}

\newcommand{\minexp}{\ensuremath{\mathrm{min\mbox{-}exp}}}

\renewcommand{\phi}{\varphi}

\newcommand{\tn}[1]{\textnormal{#1}}
\newcommand{\N}{\mathbb{N}}
\newcommand{\Z}{\mathbb{Z}}
\renewcommand{\O}{\mathcal{O}}
\newcommand{\oli}[1]{\overline{#1}}

\newcommand{\redstyle}[1]{\mathop{#1}}
\newcommand{\reduction}[2]{\redstyle{\le_{\mathrm{#2}}^{\mathrm{#1}}}}
\newcommand{\polyreduction}[1]{\reduction{p}{#1}}
\newcommand{\redm}{\polyreduction{m}}

\newcommand{\redmNP}{\reduction{NP}{m}}
\newcommand{\redlogm}{\reduction{log}{m}}

\newcommand{\opstyle}[1]{\mathrm{#1}}

\newcommand{\cNP}{\opstyle{NP}}

\newcommand{\cSigma}[1]{\opstyle{\Sigma^P_{#1}}}
\newcommand{\cPi}[1]{\opstyle{\Pi^P_{#1}}}
\newcommand{\cDelta}[1]{\opstyle{\Delta^P_{#1}}}
\newcommand{\cPSPACE}{\opstyle{PSPACE}}

\newcommand{\cEEEXPSPACE}{\opstyle{3EXPSPACE}}

\newcommand{\MC}[1]{\mathrm{MC}_{\N}(#1)}
\newcommand{\MF}[1]{\mathrm{MF}_{\N}(#1)}
\newcommand{\EC}[1]{\mathrm{EC}_{\N}(#1)}
\newcommand{\EF}[1]{\mathrm{EF}_{\N}(#1)}

\newcommand{\SC}[1]{\mathrm{SC}_{\N}(#1)}
\newcommand{\SF}[1]{\mathrm{SF}_{\N}(#1)}
\newcommand{\CSP}[2][\{\N\}]{\mathrm{CSP}(#1;#2)}

\newcommand{\com}{{^-}}

\newcommand{\free}{\star}

\newcommand{\ignore}[1]{}

\newcommand{\bA}{\mathfrak{A}}

\begin{document}

\title{Constraint Satisfaction Problems around Skolem Arithmetic}
\author{Christian Gla{\ss}er\inst{1} \and Peter Jonsson\inst{2}\thanks{The second author was partially supported by the {\em Swedish
Research Council} (VR) under grant 621-2012-3239.} \and Barnaby Martin\inst{3}\thanks{The third author was supported by EPSRC grant EP/L005654/1.}}
\institute{Theoretische Informatik, Julius-Maximilians-Universit\"at, W\"urzburg, Germany \and \mbox{Dept.} of Computer and Information Science, 
Link\"opings Universitet, 
SE-581 83 Link\"oping, 
Sweden \and School of Science and Technology, Middlesex University,\\
The Burroughs, Hendon, London NW4 4BT\\
%\email{barnabymartin@gmail.com}
}

%\author{Petar Dapi\'c\inst{1} \and Petar Markovi\'c\inst{1} \and Barnaby Martin\inst{2}}
%\institute{Departman za matematiku i informatiku, University of Novi Sad, Serbia \and School of Science and Technology, Middlesex University,\\
%The Burroughs, Hendon, London NW4 4BT, U.K.\\
%%\email{barnabymartin@gmail.com}
%}

\maketitle

\begin{abstract}
We study interactions between Skolem Arithmetic and certain classes of Constraint Satisfaction Problems (CSPs). We revisit results of Gla\ss er et \mbox{al.} \cite{GlasserEtAl2010} in the context of CSPs and settle the major open question from that paper, finding a certain satisfaction problem on circuits to be decidable. This we prove using the decidability of Skolem Arithmetic.
We continue by studying first-order expansions of Skolem Arithmetic without constants, $(\mathbb{N};\times)$, as CSPs. We find already here a rich landscape of problems with non-trivial instances that are in P as well as those that are NP-complete.
\end{abstract}
  
\section{Introduction}

A \emph{constraint satisfaction problem} (CSP) is a computational problem in which the input consists of a finite set of variables and a finite set of constraints, and where the question is whether there exists a mapping from the variables to some fixed domain such that all the constraints are satisfied. 
When the domain is finite, and arbitrary constraints are permitted in the input, the CSP is NP-complete. 
However, when only constraints from a restricted set of relations are allowed in the input, it can be possible to solve the CSP in polynomial time. 
The set of relations that is allowed to formulate the constraints in the input is often called the \emph{constraint language}. The question which constraint
languages give rise to polynomial-time solvable CSPs
has been the topic of intensive research over the past years. It has been conjectured by Feder and Vardi~\cite{FederVardi} that CSPs for constraint languages over finite domains have a complexity dichotomy: they are either in P or NP-complete.  This conjecture remains unsettled, although dichotomy is now known on substantial classes (e.g. structures of size $\leq 3$ \cite{Schaefer,Bulatov} and smooth digraphs \cite{HellNesetril,barto:1782}). Various methods, combinatorial (graph-theoretic), logical and universal-algebraic have been brought to bear on this classification project, with many remarkable consequences. A conjectured delineation for the dichotomy was given in the algebraic language in \cite{JBK}.

By now the literature on infinite-domain CSPs is also beginning to mature. Here the complexity can be much higher (\mbox{e.g.} undecidable) but on natural classes there is often the potential for structured classifications, and this has proved to be the case for reducts of, \mbox{e.g.} the rationals with order \cite{tcsps-journal}, the random (Rado) graph \cite{BodPin-Schaefer} and the integers with successor \cite{dCSPs2}; as well as first-order (fo) expansions of linear program feasibility \cite{essentially-convex}. Skolem Arithmetic, which we take here to be the non-negative integers with multiplication (and possibly constants), represents a perfect candidate for continuation in this vein. These natural classes have the property that their CSPs sit in NP and a topic of recent interest for the second and third authors has been natural CSPs sitting in higher complexity classes.

Meanwhile, a literature existed on satisfiability of circuit problems over sets of integers involving work of the first author \cite{GlasserEtAl2010}, itself continuing a line of investigation begun in \cite{StockmeyerMeyer1973} and pursued in \cite{Wagner1984,Yang2001,McKenzieWagner2007}. The problems in \cite{GlasserEtAl2010} can be seen as variants of certain functional CSPs whose domain is all singleton sets of the non-negative integers and whose relations are set operations of the form: complement, intersection, union, addition and multiplication (the latter two are defined set-wise, \mbox{e.g.} $A \times B:= \{ab : a \in A \wedge b \in B\}$). An open problem was the complexity of the problem when the permitted set operators were precisely complement, intersection, union and multiplication. In this paper we resolve that this problem is in fact decidable, indeed in triple exponential space. We prove this result by using the decidability of the theory of Skolem Arithmetic with constants. In studying this problem we are able to bring to light existing results of \cite{GlasserEtAl2010} as results about their related CSPs, providing natural examples with interesting super-NP complexities.

In the second part of the paper, Skolem Arithmetic takes centre stage as we initiate the study of the computational complexity of the CSPs of its reducts, \mbox{i.e.} those constraint languages whose relations have a fo-definition in $(\mathbb{N};\times)$. CSP$(\mathbb{N};\times)$ is in P, indeed it is trivial. The object therefore of our early study is its fo-expansions. We show that CSP$(\mathbb{N};+,\neq)$ is NP-complete, as is CSP$(\mathbb{N};\times,c)$ for each $c>1$. We further show that CSP$(\mathbb{N};\times,U)$ is NP-complete when $U$ is any non-empty set of integers greater than $1$ such that each has a prime factor $p$, for some prime $p$, but omits the factor $p^2$. Clearly, CSP$(\mathbb{N};\times,U)$ is in P (and is trivial) if $U$ contains $0$ or $1$. As a counterpoint to our NP-hardness results, we prove that CSP$(\mathbb{N};\times,U)$ is in P whenever there exists $m>1$ so that $U \supseteq \{m,m^2,m^3,\ldots\}$.

This paper is organised as previewed in this introduction. Several proofs are deferred to the appendix for reasons of space restriction.

\textbf{Related work}. Apart from the research on circuit problems mentioned above
there has been work on other variants like
circuits over integers \cite{tra06} and positive natural numbers \cite{bre07},
equivalence problems for circuits \cite{ghrtw10},
functions computed by circuits \cite{pd09}, and
equations over sets of natural numbers \cite{jo11,jo14}.

\section{Preliminaries}

Let $\mathbb{N}$ be the set of non-negative integers, and let $\mathbb{N}^+$ be the set of positive integers. For $m \in \mathbb{N}$, let $\mathrm{Div}_m$ be the set of factors of $m$. Finally, let $\{\mathbb{N}\}$ be the set of singletons $\{\{x\}:n \in \mathbb{N}\}$. In this paper we use a version of the CSP permitting both relations and functions (and constants). Thus, a \emph{constraint language} consists of a domain together with functions, relations and constants over that domain. One may thus consider a constraint language to be a first-order structure. A constraint language is a \emph{core} if all of its endomorphisms are embeddings (equivalently, if the domain is finite, automorphisms). The functional version of the CSP has previously been seen in, \mbox{e.g.}, \cite{FederMadelaineStewart}. For a purely functional constraint language, a \emph{primitive positive} sentence is the existential quantification of a conjunction of term equalities. More generally, and when relations present, we may have positive atoms in this conjunction. The problem CSP$(\Gamma)$ takes as input a primitive positive sentence $\phi$, and asks whether it is true on $\Gamma$. We will allow that the functions involved on $\phi$ be defined on a larger domain than the domain of $\Gamma$. This is rather \emph{unheimlich}\footnote{Weird. Thus spake Lindemann about Hilbert's non-constructive methods in the resolution of Gordon's problem (see \cite{Smorynski77}).} but it allows the problems of \cite{GlasserEtAl2010} to be more readily realised in the vicinity of CSPs. For example, one such typical domain is $\{\mathbb{N}\}$, but we will allow functions such as $-$ (complement), $\cup$ (union) and $\cap$ (intersection) whose domain and range is the set of all subsets of $\mathbb{N}$. We will also recall the operations of set-wise addition  $A + B:= \{a+b : a \in A \wedge b \in B\}$ and multiplication  $A \times B:= \{ab : a \in A \wedge b \in B\}$.

$\cSigma{i}$, $\cPi{i}$, and $\cDelta{i}$ denote
levels of the polynomial-time hierarchy,
while $\Sigma_i$, $\Pi_i$, and $\Delta_i$ denote
levels of the arithmetical hierarchy.
Moreover, we use the classes $\cNP = \cSigma{1}$,
$\cPSPACE = \bigcup_{k \geq 1} \tn{DSPACE}(n^k)$, and
$$\tn{3EXPSPACE} = \bigcup_{k \geq 1} \tn{DSPACE}\left(2^{2^{2^{n^k}}}\right).$$ For more on these complexity classes we refer the reader to \cite{Papa}.

For sets $A$ and $B$ we say that
$A$ is {\em polynomial-time many-one reducible} to $B$, in symbols $A \redm B$,
if there exists a polynomial-time computable function $f$ such that for all $x$
it holds that $(x \in A \Longleftrightarrow f(x) \in B)$.
If $f$ is even computable in logarithmic space,
then $A$ is {\em logspace many-one reducible} to $B$, in symbols $A \redlogm B$.
$A$ is {\em nondeterministic polynomial-time many-one reducible} to $B$,
in symbols $A \redmNP B$,
if there is a nondeterministic Turing transducer $M$ that runs in polynomial time
such that for all $x$ it holds that $x \in A$ if and only if
there exists a $y$ computed by $M$ on input $x$ with $y \in B$.
The reducibility notions $\redm$, $\redlogm$, and $\redmNP$ are transitive
and $\cNP$ is closed under these reducibilities.

%Arithmetical operations $+$ and $\cdot$ are extended to sets as follows: $M+N=\{m+n ~|~ m \in M \tn{ and } n\in N\}$ and
%$M\times N = \{m \cdot n ~|~ m \in M \tn{ and } n \in N\}$.
A {\em circuit} $C = (V,E,g_C)$ is a finite, non-empty, directed, acyclic graph
$(V,E)$ with a specified node $g_C \in V$.
The graph can contain multi-edges, it does not have
to be connected, and $V = \{1,2, \ldots, n\}$ for some $n \in \N$.
The nodes in the graph $(V,E)$ are topologically ordered, i.e., for all
$v_1,v_2 \in V$, if $v_1 < v_2$, then there is no path from $v_2$ to $v_1$.
Nodes are also called {\em gates}. Nodes with indegree $0$ are called
{\em input gates} and $g_C$ is called the {\em output gate}.
If there is an edge from gate $u$ to gate $v$,
then we say that $u$ is a {\em predecessor} of $v$ and
$v$ is the {\em successor} of $u$.

Let $\O \subseteq \{\cup,\cap,\com,+,\times\}$.
An {\em $\O$-circuit with unassigned input gates}
$C = (V,E,g_C,\alpha)$ is a circuit $(V,E,g_C)$ whose gates are labeled
by the labeling function $\alpha : V \rightarrow \O \cup \N \cup \{\free\}$
such that the following holds:
Each gate has an indegree in $\{0,1,2\}$, gates with indegree $0$
have labels from $\N \cup \{\free\}$, gates with indegree $1$ have label $\com$,
and gates with indegree $2$ have labels from $\{\cup,\cap,+,\times\}$.
Input gates with a label from $\N$ are called {\em assigned} (or constant)
input gates; input gates with label $\free$ are called {\em unassigned} (or variable)
input gates.
An {\em $\O$-formula} is an $\O$-circuit that only contains nodes
with outdegree one.

Let $u_1 < \cdots < u_n$ be the unassigned inputs in $C$
and $x_1,\ldots,x_n\in\N$.
By assigning value $x_i$ to the input $u_i$,
we obtain an {\em $\O$-circuit} $C(x_1,\ldots,x_n)$
whose input gates are all assigned.
In this circuit, each gate $g$ computes the following set $I(g)$:
If $g$ is an assigned input gate where $\alpha(g)\not=\free$,
then $I(g)=\{\alpha(g)\}$.
If $g=u_k$ is an unassigned input gate, then $I(g)=\{x_k\}$.
If $g$ has label $\com$ and predecessor $g_1$, then $I(g) = \N - I(g_1)$.
If $g$ has label $\circ \in \{\cup,\cap,+,\times\}$ and predecessors
$g_1$ and $g_2$, then $I(g) = I(g_1) \circ I(g_2)$.
Finally, let $I(C(x_1,\ldots,x_n)) = I(g_C)$ be the set computed by the
circuit $C(x_1,\ldots,x_n)$.

\begin{definition}[membership, equivalence, and satisfiability problems of circuits and formulas]
    \label{def_circuit_problems}
    ~\\
    Let $\O \subseteq \{\cup,\cap,\com,+,\times\}$.
    \begin{align*}
        \MC{\O} &= \{(C,b) ~|~
        \parbox[t]{110mm}{
            $C$ is an $\O$-circuit without unassigned inputs and $b \in I(C) \}$
        } \\
        \EC{\O} &= \{(C_1,C_2) ~|~
        \parbox[t]{110mm}{
            $C_1$ and $C_2$ are $\O$-circuits without unassigned inputs and \\
            we have $I(C_1)=I(C_2)\}$
        } \\
        \SC{\O} &= \{(C,b) ~|~
        \parbox[t]{110mm}{
            $C$ is an $\O$-circuit with unassigned inputs $u_1 < \cdots < u_n$
            and\\ there exist $x_1,\dots,x_n \in \N$ such that
            $b \in I\big(C(x_1,\ldots,x_n)\big)\}$
        }
    \end{align*}
    $\MF{\O}$, $\EF{\O}$, and $\SF{\O}$ are the variants that
    deal with $\O$-formulas instead of $\O$-circuits.
\end{definition}

When an $\O$-circuit is used as input for an
algorithm, then we use a suitable encoding such that it is possible to verify
in deterministic logarithmic space whether a given string encodes a valid
circuit.

In Section~\ref{sec:circuits}, for $i \in \mathbb{N}$, we often identify $\{i\}$ with $i$, where this can not cause a harmful confusion.

\section{Circuit Satisfiability and functional CSPs}
\label{sec:circuits}

We investigate the computational complexity of functional CSPs.
In many cases we can translate known lower and upper bounds for
membership, equivalence, and satisfiability problems of arithmetic circuits
\cite{mw07,ghrtw10,grtw10} to CSPs.
Our main result is the decidability of
$\SC{\com,\cup,\cap,\times}$ and $\CSP{\com,\cup,\cap,\times}$,
which solves the main open question of the paper \cite{grtw10}.
Table~\ref{tablecsp} summarizes the results obtained in this section and
shows open questions.
In particular, we would like to improve the gap between the
lower and upper bounds for $\CSP{\O}$,
where $\O$ contains $\cup$ and exactly one arithmetic operation ($+$ or $\times$).

%%%%%%%%%%%%%%%%%%%%%%%%%%%%%%%%%%%%%%%%%%%%%%%%%%%%%%%%%%%%%%%%%%%%%%%%%%%

We start with the observation that the equivalence of arithmetic formulas
reduces to functional CSPs. This yields several lower bounds for the CSPs.

\begin{proposition} \label{prop_5147}
    For $\O \subseteq \{\com,\cup,\cap,+,\times\}$ it holds that
    $\EF{\O} \redlogm \CSP{\O}$.
\end{proposition}
%\begin{proof}
%    An $\EF{\O}$-instance $(F_1,F_2)$ is mapped
%    to the $\CSP{\O}$-instance $F_1=F_2$.
%\end{proof}

\begin{corollary} \label{coro_5147} \
    \begin{enumerate}
        \item $\CSP{\{\com,\cup,\cap,+\}}$ and $\CSP{\{\com,\cup,\cap,\times\}}$
        are $\redlogm$-hard for $\cPSPACE$.
        \item $\CSP{\{\cup,\cap,+\}}$, $\CSP{\{\cup,\cap,\times\}}$,
        $\CSP{\{\cup,+\}}$, and \\ $\CSP{\{\cup,\times\}}$
        are $\redlogm$-hard for $\cPi{2}$.
    \end{enumerate}
\end{corollary}
%\begin{proof}
%    The statements follow from Proposition~\ref{prop_5147}
%    and the following facts \cite{ghrtw10}:
%    $\EF{\{\com,\cup,\cap,+\}}$ and $\EF{\{\com,\cup,\cap,\times\}}$
 %   are $\redlogm$-complete for $\cPSPACE$.

%\noindent   $\EF{\{\cup,\cap,+\}}$, $\EF{\{\cup,\cap,\times\}}$,
%    $\EF{\{\cup,+\}}$, and $\EF{\{\cup,\times\}}$
%    are $\redlogm$-complete for $\cPi{2}$.
%\end{proof}

%%%%%%%%%%%%%%%%%%%%%%%%%%%%%%%%%%%%%%%%%%%%%%%%%%%%%%%%%%%%%%%%%%%%%%%%%%%

CSPs with $+$ and $\times$ can express diophantine equations,
which implies the Turing-hardness of such CSPs.

\begin{proposition} \label{prop_73478}
    $\CSP{+,\times}$ is $\redlogm$-hard for $\Sigma_1$.
\end{proposition}

\begin{proposition} \label{prop_83651}
    $\CSP{\cup,\cap,+,\times} \in \Sigma_1$.
\end{proposition}
%\begin{proof}
%    It is decidable whether a given assignment satisfies
%    a $\CSP{\cup,\cap,+,\times}$-instance.
%    Hence testing the existence of a satisfying assignment is in $\Sigma_1$.
%\end{proof}

\begin{proposition} \label{prop_67126}
    $\CSP{\com,\cup,\cap,+,\times} \in \Sigma_2$.
\end{proposition}

%%%%%%%%%%%%%%%%%%%%%%%%%%%%%%%%%%%%%%%%%%%%%%%%%%%%%%%%%%%%%%%%%%%%%%%%%%%

We show that the decidability of Skolem arithmetic \cite{FR79} can be used to
decide the satisfiability of arithmetic circuits without $+$.
This solves the main open question of the paper \cite{grtw10}
and at the same time implies the decidability of corresponding CSPs.

\begin{theorem} \label{thm_3expspace}
    $\SC{\com,\cup,\cap,\times} \in \tn{3EXPSPACE}$.
\end{theorem}
\begin{proof}
    Let $C = C(x_1, \ldots, x_n)$ be a circuit with gates $g_1, \ldots, g_r$,
    where $g_1, \ldots, g_n$ are the input gates and $g_r$ is the output gate.
    Without loss of generality we may assume that $C$ does not have $\cap$-gates.
    For every gate $g_k$ we define a formula
    $\varphi_k := \varphi_k(x_1, \ldots, x_n, i_k, v_k, b_k)$
    in Skolem arithmetic such that the following holds.

\vspace{0.2cm}
\noindent ($*$)  For $a_1, \ldots, a_n, v \in \N$, $b \in \{0,1\}$, and $i = 1, \ldots, k$
        it holds that
        $\varphi_k(a_1, \ldots, a_n, 0,$ $v, b)$ is true and \vspace{-.2cm}
\begin{itemize}
\item $\varphi_k(a_1, \ldots, a_n, i, v, b)$ is true IFF
\item ($b=1$ iff $C(a_1, \ldots, a_n)$ produces at $g_i$ a set that contains $v$).
\end{itemize}

    Let $\varphi_0 := b_0 \vee \neg b_0 \vee (x_1 \cdot \ldots \cdot x_n \cdot i_0 \cdot v_0 = 0)$, which is always true.
    For $1 \le k \le n$, the formula $\varphi_k$ which corresponds to the $k$-th
    input gate is defined as
    $$\varphi_k :=
    \parbox[t]{140mm}{
        $\exists i_{k-1}, v_{k-1}, b_{k-1}$\\
        $[(i_k=k \wedge b_k=0) \rightarrow (x_k \neq v_k \wedge i_{k-1}=0)] \wedge$\\
        $[(i_k=k \wedge b_k=1) \rightarrow (x_k = v_k \wedge i_{k-1}=0)] \wedge$\\
        $[i_k \neq k \rightarrow (i_{k-1}=i_k \wedge v_{k-1}=v_k \wedge b_{k-1}=b_k)] \wedge$\\
        $\varphi_{k-1}$.
    }$$
    Observe that the free variables of $\varphi_k$ are variables are $x_1, \ldots, x_n, i_k, v_k, b_k$,
    i.e., $\varphi_k =$  $\varphi_k(x_1, \ldots, x_n, i_k, v_k, b_k)$.
    Moreover, an induction on $k$ shows that ($*$) holds for all $\varphi_k$
    where $0 \le k \le n$.

    Now define the formulas $\varphi_k$ for the inner gates $g_k$
    where $n < k \le r$.
    Here $d_k$, $e_k$, $f_k$, $f_k'$, $h_k$, and $h_k'$ are used as
    auxiliary variables.

    If $g_k$ is a complement gate with predecessor $g_p$, then let
    $$\varphi_k :=
    \parbox[t]{140mm}{
        $\exists i_{k-1}, v_{k-1}, b_{k-1}$\\
        $[i_k=k \rightarrow (i_{k-1}=p \wedge v_{k-1}=v_k \wedge (b_k=1 \rightarrow b_{k-1}=0) \wedge (b_k=0 \rightarrow b_{k-1}=1))] \wedge$\\
        $[i_k \neq k \rightarrow (i_{k-1}=i_k \wedge v_{k-1}=v_k \wedge b_{k-1}=b_k)] \wedge\\
        \varphi_{k-1}$.
    }$$

    If $g_k$ is a $\cup$-gate with predecessors $g_p$ and $g_q$, then let
    $$\varphi_k :=
    \parbox[t]{140mm}{
        $\exists f_k, h_k \forall e_k \exists i_{k-1}, v_{k-1}, b_{k-1}$\\
        $[(i_k=k \wedge e_k=0) \rightarrow (i_{k-1}=p \wedge v_{k-1}=v_k \wedge b_{k-1}=f_k)] \wedge$\\
        $[(i_k=k \wedge e_k \neq 0) \rightarrow (i_{k-1}=q \wedge v_{k-1}=v_k \wedge b_{k-1}=h_k)] \wedge$\\
        $[(i_k=k \wedge b_k=1) \rightarrow (f_k=1 \vee h_k=1)] \wedge$\\
        $[(i_k=k \wedge b_k=0) \rightarrow (f_k=0 \wedge h_k=0)] \wedge$\\
        $[i_k \neq k \rightarrow (i_{k-1}=i_k \wedge v_{k-1}=v_k \wedge b_{k-1}=b_k )] \wedge$\\
        $\varphi_{k-1}$.
    }$$

    If $g_k$ is a $\times$-gate with predecessors $g_p$ and $g_q$, then let
    $$\varphi_k :=
    \parbox[t]{140mm}{
        $\exists f_k, f_k' \forall e_k \forall h_k, h_k' \exists d_k \exists i_{k-1}, v_{k-1}, b_{k-1}$\\
        $[(i_k=k \wedge b_k=1 \wedge e_k=0) \rightarrow (f_k \cdot f_k' = v_k \wedge i_{k-1}=p \wedge v_{k-1}=f_k \wedge b_{k-1}=1)] \wedge$\\
        $[(i_k=k \wedge b_k=1 \wedge e_k \neq 0) \rightarrow (f_k \cdot f_k' = v_k \wedge i_{k-1}=q \wedge v_{k-1}=f_k' \wedge b_{k-1}=1)] \wedge$\\
        $[(i_k=k \wedge b_k=0 \wedge h_k \cdot h_k' = v_k \wedge d_k=0) \rightarrow (i_{k-1}=p \wedge v_{k-1}=h_k \wedge b_{k-1}=0)] \wedge$\\
        $[(i_k=k \wedge b_k=0 \wedge h_k \cdot h_k' = v_k \wedge d_k \neq 0) \rightarrow (i_{k-1}=q \wedge v_{k-1}=h_k' \wedge b_{k-1}=0)] \wedge$\\
        $[i_k \neq k \rightarrow (i_{k-1}=i_k \wedge v_{k-1}=v_k \wedge b_{k-1}=b_k )] \wedge$\\
        $\varphi_{k-1}$.
    }$$

    Again it holds that $\varphi_k$'s free variables are
    $x_1, \ldots, x_n, i_k, v_k, b_k$ and an induction on $k$
    shows that ($*$) holds for all $\varphi_k$ where $0 \le k \le r$.
    So for the output gate $g_r$ we obtain
    $$(C,v) \in \SC{\com,\cup,\cap,\times)}
    \;\Longleftrightarrow\;
    \exists a_1, \ldots, a_n \; \varphi_r(a_1, \ldots, a_n, r, v, 1).$$
    The right-hand side is a first-order sentence of Skolem arithmetic.
    On input $(C,v)$ this sentence can be computed in polynomial time,
    which shows that $\SC{\com,\cup,\cap,\times}$ is $\redm$-reducible
    to Skolem arithmetic.
    The latter is decidable in $\tn{3EXPSPACE}$ \cite{FR79}.
\end{proof}

\begin{corollary} \label{coro_76612}
    $\CSP{\com,\cup,\cap,\times} \in \tn{3EXPSPACE}$
\end{corollary}
\begin{proof}
    By Theorem~\ref{thm_3expspace}, it suffices to show
    $\CSP{\com,\cup,\cap,\times} \redm \SC{\com,\cup,\cap,\times}$.
    We describe the reduction on input of a $\CSP{\com,\cup,\cap,\times}$-instance
    $x = \exists y \in \N^n \bigwedge_{i=0}^m (t_{2i}=t_{2i+1})$.
    Observe that
    \begin{eqnarray*}
        \bigwedge_{i=0}^m (t_{2i}=t_{2i+1}) &\Longleftrightarrow&
        \bigwedge_{i=0}^m (t_{2i} \cap \overline{t_{2i+1}}) \cup
        (\overline{t_{2i}} \cap t_{2i+1})=\emptyset\\
        &\Longleftrightarrow&
        \bigcup_{i=0}^m [(t_{2i} \cap \overline{t_{2i+1}}) \cup
        (\overline{t_{2i}} \cap t_{2i+1})] \, = \, \emptyset\\
        &\Longleftrightarrow&
        0 \; \in \; \underbrace{\overline{0 \cdot
        \bigcup_{i=0}^m [(t_{2i} \cap \overline{t_{2i+1}}) \cup
        (\overline{t_{2i}} \cap t_{2i+1})]}}_{C :=}.
    \end{eqnarray*}
    So $x \in \CSP{\com,\cup,\cap,\times}$ if and only if
    $(C,0) \in \SC{\com,\cup,\cap,\times}$.
\end{proof}

\begin{corollary} \label{coro_9843}
    $\CSP{\com,\cup,\cap,+} \in \tn{3EXPSPACE}$
\end{corollary}
\begin{proof}
    By Corollary~\ref{coro_76612}, it suffices to show that we have
    $\CSP{\com,\cup,\cap,+} \redm$ $\CSP{\com,\cup,\cap,\times}$.
    Consider a $\CSP{\com,\cup,\cap,+}$-instance
    $$x := \exists y \in \N^n \bigwedge_{i=0}^m (t_{2i}=t_{2i+1}).$$

    We may assume that $0$ and $1$ are the only constants that occur in $x$.
    We can do this, since constants $c>1$ can be removed as follows:
    Let $l=\lfloor \log c \rfloor$,
    replace $c$ with a new variable $z$, and add constraints
    $$(z_0=\{2\}) \wedge (z_1=z_0 + z_0) \wedge \cdots \wedge
    (z_l=z_{l-1} + z_{l-1}) \wedge (z = \sum_{i \in I} z_i),$$
    where $z_0, \ldots, z_l$ are new variables and
    \[I = \{ i ~|~ \tn{the $i$-th bit in $c$'s binary representation is $1$}\}.\]
    Note that removing constants in this way can be done in polynomial time.

    Observe that the term
    $q := \oli{\left(\oli{\oli{\{0,1\}} \cdot \oli{\{0,1\}}} \cap
    \oli{\{0,1,2\}}\right) \cdot \oli{\{0\} \cap \{1\}}}$
    generates the set $\{2^i ~|~ i \in \N\}$.

    For every term $t$, let $t'$ be the term that is obtained from $t$ if
    every constant $c$ is replaced with $2^c$,
    every $+$ operation is replaced with $\times$, and
    every complement operation $\oli{s}$ is replaced with $(\oli{s} \cap q)$.
    The computation of $t'$ is possible in polynomial time,
    since only the constants $0$ and $1$ can appear.

    The reduction outputs the $\CSP{\com,\cup,\cap,\times}$-instance
    $$x' := \exists y \in \N^n \bigwedge_{i=0}^m (t'_{2i}=t'_{2i+1}) \wedge
    \bigwedge_{i=1}^n (y_i \cup q = q).$$

    Observe that for all terms $t$ and all $e = (e_1, \ldots, e_n) \in \N^n$
    it holds that
    \begin{equation} \label{eqn_8239}
        t'(2^{e_1}, \ldots, 2^{e_n}) = \{ 2^i ~|~ i \in t(e_1, \ldots, e_n)\}.
    \end{equation}

    It remains to show that $x$ and $x'$ are equivalent.

    If $e = (e_1, \ldots, e_n) \in \N^n$ is a satisfying assignment for $x$,
    then by equation~(\ref{eqn_8239}),
    $z = (2^{e_1}, \ldots, 2^{e_n})$ is a satisfying assignment for $x'$
    (note that $\bigwedge_{i=1}^n (y_i \cup q = q)$ holds,
    since $y_i = 2^{e_i} \in q$).

    If $z = (z_1, \ldots, z_n) \in \N^n$ is a satisfying assignment for $x'$,
    then because of the constraints $\bigwedge_{i=1}^n (y_i \cup q = q)$,
    $z_1 = 2^{e_1}, \ldots, z_n = 2^{e_n}$ for $e = (e_1, \ldots, e_n) \in \N^n$
    and by (\ref{eqn_8239}), $e$ is a satisfying assignment for $x$.
\end{proof}

%%%%%%%%%%%%%%%%%%%%%%%%%%%%%%%%%%%%%%%%%%%%%%%%%%%%%%%%%%%%%%%%%%%%%%%%%%%

The following propositions transfer the $\cNP$-hardness
from satisfiability problems for arithmetic circuits to
$\CSP{\times}$ and $\CSP{+}$.

\begin{proposition} \label{prop_6286}
    $\CSP{\times}$ is $\redlogm$-hard for $\cNP$.
\end{proposition}

\begin{proposition} \label{prop_6287}
    $\CSP{+}$ is $\redlogm$-hard for $\cNP$.
\end{proposition}

%%%%%%%%%%%%%%%%%%%%%%%%%%%%%%%%%%%%%%%%%%%%%%%%%%%%%%%%%%%%%%%%%%%%%%%%%%%

The remaining results in this section show
that certain functional CSPs belong to $\cNP$.
This needs non-trivial arguments of the form:
If a CSP can be satisfied, then even with small values.
These arguments are provided by the known results that
integer programs, existential Presburger arithmetic,
and existential Skolem arithmetic are decidable in $\cNP$
\cite{bt76,pap81,Sca84,gra89}.

\begin{proposition} \label{prop_12481}
    $\CSP{\com,\cap,\cup}$ is $\redlogm$-complete for $\cNP$.
\end{proposition}

\begin{proposition} \label{prop_98127}
    $\CSP{+} \in \cNP$.
\end{proposition}
\begin{proof}
    Consider a $\CSP{+}$-instance
    $\varphi := \exists x_1, \ldots, x_n
    [s_1 = t_1 \wedge \cdots \wedge s_m = t_m]$.
    Each atom $s_i=t_i$ term can be written as
    $0 = t_i-s_i = a_{i,1} x_1 + \cdots + a_{i,n} x_n$ where $a_{i,j} \in \Z$.
    Let 
    $$A = 
    \left(
    \begin{matrix}
        a_{1,1} & \cdots & a_{1,m}\\
        \vdots & & \vdots \\
        a_{n,1} & \cdots & a_{n,m}
    \end{matrix}
    \right).
    $$
    Hence $\varphi \in \CSP{+}$ if and only if
    there exists an $x = (x_1, \ldots, x_n) \in \N^n$ such that $Ax = 0$.
    The right-hand side is an integer program that can be solved in $\cNP$
    \cite{bt76,pap81}.
\end{proof}

\begin{proposition} \label{prop_981281} \
    \begin{enumerate}
        \item $\CSP{\cap,+} \redmNP \CSP{+,=,\neq}$.
        \item $\CSP{\cap,\times} \redmNP \CSP{\times,=,\neq}$.
    \end{enumerate}
\end{proposition}

\begin{corollary} \label{coro_981281}
    $\CSP{\cap,+}, \CSP{\cap,\times} \in \cNP$.
\end{corollary}
%\begin{proof}
%    $\CSP{+,=,\neq}$-instances and $\CSP{\times,=,\neq}$-instances are formulas of
%    existential Presburger arithmetic and existential Skolem arithmetic,
%    which are both decidable in $\cNP$ \cite{Sca84,gra89}.
%    Now the statement follows from Proposition~\ref{prop_981281}.
%\end{proof}

%%%%%%%%%%%%%%%%%%%%%%%%%%%%%%%%%%%%%%%%%%%%%%%%%%%%%%%%%%%%%%%%%%%%%%%%%%%

\begin{center}

\begin{table}[hbt]
\renewcommand{\arraystretch}{1.3}
\[
\begin{array}{|@{\ }c@{\ }c@{\ }c@{\ }c@{\ }c@{\ }|cc|cc@{\ }|}
\hline
\multicolumn{5}{|c|}{} & \multicolumn{4}{|c|}{\tn{$\CSP{\O}$}} \\
\hline
\multicolumn{5}{|c|}{\O}& \multicolumn{2}{|c|}{\textrm{~~~~~~~~~~~~~~~Lower Bound~~~~~~~~~~~~~~}} & \multicolumn{2}{|c|}{\textrm{~~~~~~~~~~~~~~~Upper Bound~~~~~~~~~~~~~~}} \\
\hline
\com    & \cup  & \cap  & +     & \times    & \Sigma_1          & \tn{P\ref{prop_73478}}    & \Sigma_2          & \tn{P\ref{prop_67126}}    \\
\hline
\com    & \cup  & \cap  & +     &           & \cPSPACE          & \tn{C\ref{coro_5147}}     & \cEEEXPSPACE      & \tn{C\ref{coro_9843}} \\
\hline
\com    & \cup  & \cap  &       & \times    & \cPSPACE          & \tn{C\ref{coro_5147}}     & \cEEEXPSPACE      & \tn{C\ref{coro_76612}} \\
\hline
\com    & \cup  & \cap  &       &           & \cNP              & \tn{P\ref{prop_12481}}    & \cNP              & \tn{P\ref{prop_12481}}   \\
\hline
        & \cup  & \cap  & +     & \times    & \Sigma_1          & \tn{P\ref{prop_73478}}    & \Sigma_1          & \tn{P\ref{prop_83651}} \\
\hline
        & \cup  & \cap  & +     &           & \cPi{2}           & \tn{C\ref{coro_5147}}     & \cEEEXPSPACE      & \tn{C\ref{coro_9843}} \\
\hline
        & \cup  & \cap  &       & \times    & \cPi{2}           & \tn{C\ref{coro_5147}}     & \cEEEXPSPACE      & \tn{C\ref{coro_76612}} \\
\hline
%        & \cup  & \cap  &       &           & \tn{trivial}      &                           & \tn{trivial}      &                   \\
%\hline
        & \cup  &       & +     & \times    & \Sigma_1          & \tn{P\ref{prop_73478}}    & \Sigma_1          & \tn{P\ref{prop_83651}} \\
\hline
        & \cup  &       & +     &           & \cPi{2}           & \tn{C\ref{coro_5147}}     & \cEEEXPSPACE      & \tn{C\ref{coro_9843}} \\
\hline
        & \cup  &       &       & \times    & \cPi{2}           & \tn{C\ref{coro_5147}}     & \cEEEXPSPACE      & \tn{C\ref{coro_76612}} \\
\hline
%        & \cup  &       &       &           & \tn{trivial}      &                           & \tn{trivial}      &                   \\
%\hline
        &       & \cap  & +     & \times    & \Sigma_1          & \tn{P\ref{prop_73478}}    & \Sigma_1          & \tn{P\ref{prop_83651}} \\
\hline
        &       & \cap  & +     &           & \cNP              & \tn{P\ref{prop_6287}}     & \cNP              & \tn{C\ref{coro_981281}}    \\
\hline
        &       & \cap  &       & \times    & \cNP              & \tn{P\ref{prop_6286}}     & \cNP              & \tn{C\ref{coro_981281}} \\
\hline
%        &       & \cap  &       &           & \tn{trivial}      &                           & \tn{trivial}      &                   \\
%\hline
        &       &       & +     & \times    & \Sigma_1          & \tn{P\ref{prop_73478}}    & \Sigma_1          & \tn{P\ref{prop_83651}} \\
\hline
        &       &       & +     &           & \cNP              & \tn{P\ref{prop_6287}}     & \cNP              & \tn{P\ref{prop_98127}}    \\
\hline
        &       &       &       & \times    & \cNP              & \tn{P\ref{prop_6286}}     & \cNP              & \tn{C\ref{coro_981281}} \\
\hline
\end{array}
\]
\caption{Upper and lower bounds for $\CSP{\O}$.
All lower bounds are with respect to $\redlogm$-reductions.} \label{tablecsp}
\end{table}
\end{center}

\section{CSPs over fo-expansions of Skolem Arithmetic}
\label{sec:reducts-of-Skolem}

We now commence our exploration of the complexity of CSPs generated from the simplest expansions of $(\mathbb{N};\times)$. Abandoning our set-wise definitions, we henceforth use $\times$ to refer to the syntactic multiplication of Skolem Arithmetic (which may additionally carry semantic content). When we wish to refer to multiplication in a purely semantic way, we prefer $\cdot$s or $\prod$. We will consider $\times$ as a ternary relation rather than a binary function. We will never use syntactic $\times$ in a non-standard way, \mbox{i.e.} holding on a triple of integers for which it does not already hold in natural arithmetic.

\begin{proposition}
\label{prop:in-NP}
Let $\Gamma$ be a finite signature reduct of $(\mathbb{N};\times,1,2,\ldots)$. Then CSP$(\Gamma)$ is in NP.
\end{proposition}
\begin{proof}
It is known that Skolem Arithmetic admits quantifier-elimination and that the existential theory in this language is in NP \cite{Graedel89}. The result follows when one considers that we can substitute quantifier-free definitions for each among our finite set of fo-definable relations.
\end{proof}

\subsection{Upper bounds}

We continue with polynomial upper bounds. 
\begin{lemma} \label{lem:preprocessing}
Let $U \subseteq N$ be non-empty and $U \cap \{0,1\} = \emptyset.$
CSP$({\mathbb N};\times,U)$ is polynomial-time reducible to 
CSP$({\mathbb N}^+;\times,U)$.
\end{lemma}
\begin{proof}
Let $\phi$ be an arbitrary instance of CSP$({\mathbb N};\times,U)$ involving a set of atoms $C$ on variables $V$. 
We construct a set $V' \subseteq V$ incrementally by repeating the following three steps
until a fixed point is reached.

\begin{itemize}
\item
if $U(v) \in C$, then $V':=V' \cup \{v\}$,

\item
if $(x \times y =z) \in C$ and $z \in V'$, then $V':=V' \cup \{x,y\}$, and

\item
if $(x \times y=z) \in C$ and $x,y \in V'$, then $V':=V' \cup \{z\}$.
\end{itemize}

Note that if $v \in V'$, then any solution to $\phi$ must satisfy $s(v) \neq 0$.
Let $V_0=V \setminus V'$. Construct $\phi'$, an instance for CSP$(\mathbb{N};\times,U,0)$, with atoms $C'$ and variables $V'$ by replacing each variable $v \in V_0$ with the constant $0$. Note the following:

\begin{enumerate}
\item
if $U(0) \in C'$, then the variable that was replaced by $0$ is a member of $V'$
so this case cannot occur.

\item
if $(0 \times y=z) \in C'$ or $(x \times 0=z) \in C'$, 
then the variable replaced by $0$ is not a member of $V'$ while $z$ is a member
of $V'$. This situation cannot occur.

\item
if $(x\times y=0) \in C'$, then note that $x,y \in V'$ so the variable
that was replaced with $0$ also was a member of $V'$. Hence, this case cannot
occur.
\end{enumerate}
Thus, $0$ can only appear in three cases: $(x\times 0=0)$,
$(0\times x=0)$, and $(0\times 0 = 0)$.
Let $\phi''$ be an instance for CSP$(\mathbb{N}^+;\times,U)$ built from atoms $C''$ and variables $V'$ where $C''$ is obtained from $C'$ by removing these kinds
of constraints. Note that $(\mathbb{N};\times,U,0) \models \phi'$ iff $(\mathbb{N};\times,U) \models \phi''$.
Also note that if $\phi''$ has a solution, then
it has a solution $s:V' \rightarrow {\mathbb N}^+$. 
This implies that $\phi'$ is satisfiable on $(\mathbb{N};\times,U)$ iff it is satisfiable on $(\mathbb{N}^+;\times,U)$

The transformation above can obviously be carried out in polynomial time. In order to prove
the lemma,
it remains to show that $(\mathbb{N};\times,U) \models \phi$ iff $(\mathbb{N};\times,U,0) \models \phi'$.

Assume first that $\phi$ has a solution $s:V \rightarrow {\mathbb N}$. Since $C'' \subseteq C$, it follows immediately that $s$ is a solution to $\phi''$, too.

Assume instead that $\phi''$ has a solution $s':V' \rightarrow {\mathbb N}^+$.
We claim that the function $s:V \rightarrow {\mathbb N}$ defined by
$s(v)=s'(v)$ for $v \in V'$ and $s(v)=0$ otherwise is a solution to $\phi$.
If the variable $v \in V$ appears in an atom $U(v)$, then $v \in V'$ and
$U(v)$ is satisfied by $s$. Consider an atom $(x\times y=z) \in C$.
If $\{x,y,z\} \subseteq V'$, then $s$ satisfies the atom since $(x\times y=z) \in C''$.
Assume $x \not\in V'$. Then $z \not\in V'$, $s(x)=s(z)=0$ and the atom is satisfied
by $s$. The same reasoning applies when $y \not\in V'$.
Assume finally that $z \not\in V'$. Then 
at least one of $x,y \not\in V'$ so
$s(x)=0$ and/or $s(y)=0$. Combining this with the fact that $s(z)=0$ implies that
the atom is satisfied by $s$.
\end{proof}

We now borrow the following slight simplification of Lemma 6 from \cite{JonssonLoow}.
\begin{lemma}[Scalability \cite{JonssonLoow}]
\label{lem:scalability}
Let $\Gamma$ be a finite signature constraint language with domain $\mathbb{R}$, whose relations are quantifier-free definable in $+, \leq$ and $<$, such that the following holds.
\begin{itemize}
\item Every satisfiable instance of CSP$(\Gamma)$ is satisfied by some rational point.
\item For each relation $R \in \Gamma$ , it holds that if $\overline{x}:=(x_1, x_2, \ldots, x_k) \in R$, then $(ax_1, ax_2, \ldots, ax_k) \in R$ for all $a \in \{y: y \in \mathbb{R},  y \geq 1\}$.
\item CSP$(\Gamma)$ is in P.
\end{itemize}
Then CSP$(\Delta)$ is in P, where $\Delta$ is obtained from $\Gamma$ by substituting the domain $\mathbb{R}$ by $\mathbb{Z}$.
\end{lemma}

\begin{lemma} \label{lem:LPreduction}
Arbitrarily choose $m > 1$ and $U \subseteq {\mathbb N}^+$ such that
$\{m,m^2,m^3,\ldots\} \subseteq U$. Then, CSP$({\mathbb N}^+;\times,U)$ is in P.
\end{lemma}
\begin{proof}
Define $h(1)=1$ and $h(x)=m^{\ell(x)}$ when $x > 1$. Let $D=\{1,m,m^2,m^3,\dots\}$.
The function $h$ is a homomorphism from $({\mathbb N}^+;\times,U)$ to 
$(D;\times,U \cap D)$. 
Clearly, $h(U) = U \cap D$.
Suppose $a \cdot b=c$ where $a,b,c \in {\mathbb N}^+$. 
We see that

\[h(a) \cdot h(b)=m^{\ell(a)} \cdot m^{\ell(b)} = m^{\ell(a) + \ell(b)} = m^{\ell(a)+\ell(b)}
=m^{\ell(c)} = h(c).\]

Define $h'(1)=0$ and $h'(m^k)=k$. Note that $h'$ is a homomorphism
from $(D;\times,U \cap D)$ to $({\mathbb N};+,x\geq 1)$.
We know that that CSP$({\mathbb R};+,x \geq 0,x \geq 1)$
is in P (via linear programming) and this implies tractability of
CSP$({\mathbb N};+,x \geq 1)$ through CSP$({\mathbb Z};+,x \geq 1, x \geq 0)$ by Lemma~\ref{lem:scalability}.
\end{proof}

\begin{proposition}
\label{prop:P}
Arbitrarily choose $m > 1$ and $U \subseteq {\mathbb N}$ such that
$\{m,m^2,m^3,\ldots\}$ $\subseteq U$. Then, CSP$({\mathbb N};\times,U)$ is in P.
\end{proposition}
\begin{proof}
Combine Lemma~\ref{lem:preprocessing} with Lemma~\ref{lem:LPreduction}.
\end{proof}

\subsection{Cores}

We say that an integer $m>1$ has a {\em degree-one factor} $p$ if and only if
$p$ is a prime such that $p | m$ and $p^2 \not| \; m$. Let $\mathrm{Div}_m$ be the set of divisors of $m$, pp-definable in $(\mathbb{N};\times,m)$ by $\exists y \ x \times y=m$. We can pp-define the relation $\{1\}$ in $(\mathrm{Div}_m; \times, m)$ since $x=1$ iff $x \times x=x$ (recalling $0 \notin \mathrm{Div}_m$). It follows that $\{1,m\}$ are contained in the core of $(\mathrm{Div}_m; \times, m)$.

\begin{lemma} 
Let $m > 1$ be an integer that has a degree-one factor $p$.
Then $(\mathrm{Div}_m; \times,m)$ has a two-element core.
\label{lem:2el-core}
\end{lemma}
\begin{proof}
Consider the function $e:\mathrm{Div}_m \rightarrow \mathrm{Div}_m$ uniquely defined by
$e(1)=1$, $e(p)=m$, $e(p_1)=\dots=e(p_k)=1$ (\mbox{i.e.} all the other prime divisors map to $1$), and the rule $e(x \cdot x')=e(x) \cdot e(x')$.
We claim that $e$ is an endomorphism of $(\mathrm{Div}(m); \times,m)$. Clearly, $e(m)=m$.
Arbitrarily choose a tuple $(x,y,z) \in (x \times y=z)$. 
Let $x=x_1^{\alpha_1} \cdot \ldots \cdot x_a^{\alpha_a}$ and
$y=y_1^{\beta_1} \cdot \ldots \cdot y_b^{\beta_b}$ be prime factorizations.
Note that at most one of $x_1,\ldots,x_a,y_1,\ldots,y_b$ can equal $p$ and, if so, the
corresponding exponent must equal one.
If none of the factors equal $p$, then
$e(x)=e(y)=e(z)=1$ and $e(x) \times e(y) = e(z)$.
Otherwise, assume without loss of generality that $x_1=p$.
Then we have $e(x)=e(z)=m$ and $e(y)=1$. Once again $e(x) \times e(y) = e(z)$
and $e$ is indeed an endomorphism of $(\mathrm{Div}_m; \times,m)$. 
It follows that $(\{1,m\};\times,m)$ is the core of $(\mathrm{Div}_m; \times,m)$.
\end{proof}

\begin{lemma}
Let $m$ be an integer
that does not have a degree-one factor and
let $D$ contain the divisors of $m$.
Then  $(\mathrm{Div}_m; \times,m)$ does not have a two-element core. 
\end{lemma}
\begin{proof}
Assume $m$ has the prime factorization $m=p_1^{\alpha_1} \cdot \ldots \cdot p_k^{\alpha_k}$ and
note that $\alpha_1,\ldots,\alpha_k > 1$.
Assume $e:\mathrm{Div}_m \rightarrow \{a,b\}$ is an endomorphism to a two-element core, \mbox{i.e.} the range of $e$ is $\{1,m\}$. 
Since multiplication is determined by the action of the primes, we can see that for one prime $p \in \{p_1,\ldots,p_k\}$ we must have $e(p)=m$.
Consider $p\times p=p^2$. If we apply the endomorphism $e$ to this tuple,
we end up with $e(p)\times e(p) = e(p)^2$ which is not possible. Hence, $e$ does not exist and 
$(\mathrm{Div}_m; \times,m)$ does not admit a two-element core.
\end{proof}

\subsection{Lower bounds}

We now move to lower bounds of NP-completeness.
\begin{proposition}
\label{prop:neq}
CSP$(\mathbb{N};\neq,\times)$ is NP-complete. 
\end{proposition}
\begin{proof}
NP membership follows from Proposition~\ref{prop:in-NP}. For NP-hardness, we will encode the CSP of a certain Boolean constraint language, \mbox{i.e.} with domain $\{0,1\}$, with two relations: $\neq$ and $R_1:=\{0,1\}^3\setminus \{(1,1,1)\}$. This CSP is NP-hard because $\neq$ omits constant and semilattice polymorphisms and $R_1$ omits majority and minority polymorphisms (an algebraic reformulation of Schaefer's Theorem \cite{Schaefer} in the spirit of \cite{Jeavons}).

To encode our Boolean CSP, ensure all variables $v$ satisfy $v\times v=v$, which enforces the domain $\{0,1\}$. Consider $0$ to be false and $1$ to be true. $0$ is pp-definable by $x\times x=x \wedge \exists y \ y \neq x \wedge x\times y=x$. For $\{0,1\}^3\setminus \{(1,1,1)\}$ take $R_1(x,y,z)$ to be $\exists w\ x \times y =z  \wedge w \times z=0$; and for $\neq$ take $\neq$. The reduction may now be done by local substitution and the result follows.
\end{proof}

An operation $t:D^k \rightarrow D$ is a {\em weak near-unanimity} operation if $t$ is
idempotent and satisfies the equations

\[t(y,x,\dots,x)=t(x,y,x,\dots,x)=\dots=t(x,\dots,x,y).\]

\begin{theorem}[\cite{BartoK09}]
\label{algebrahardness}
Let $\Gamma$ be a constraint language over a finite set $D$. If $\Gamma$ is a core
and does not have a weak near-unanimity polymorphism, then CSP$(\Gamma)$ is NP-hard.
\end{theorem}

\begin{lemma} \label{finitehard}
Arbitrarily choose an $m>1$ such that $m \neq k^n$ for all $k,n > 1$ together with 
a finite set $\{1,m\} \subseteq S \subseteq {\mathbb N} \setminus \{0\}$.
If $(\mathrm{Div}_m;\times,m)$ is a core, then CSP$(S;\times,m)$ is NP-hard.
\end{lemma}
\begin{proof}
Assume $(S;\times,m)$ admits a weak near-unanimity operation $t:S^k \rightarrow S$.
The relation $\prod_{i=1}^{k} x_i = x_{k+1}$ is pp-definable in $(S;\times,m)$ and so is the relation

\[R=\{(x_1,\ldots,x_k) \in S^k \; | \; \prod_{i=1}^k x_i=m\}.\]

The relation $R$ contains the tuples 
\[
\begin{array}{c}
(m,1,\ldots,1) \\ (1,m,1,\ldots,1) \\ \vdots \\ (1,\ldots,1,m).
\end{array}
\]
Applying $t$ component-wise (\mbox{i.e.} vertically) to these tuples yields a tuple $(a,\ldots,a)$ for some $a \in D$.
However, $R$ does not contain a tuple $(a,\ldots,a)$ for any $a \in D$ since this would imply that $m=a^k$
for some $a,k > 1$.
We conclude that CSP$(S;\times ,m)$ is NP-hard by Theorem~\ref{algebrahardness}.
\end{proof}
Note that the proof of this last lemma was eased by our assumption that $\times$ is a relation and not a function. Had it been a function we would have to prove the domain $S$ would be closed under it.

\begin{theorem} \label{infinitemul}
CSP$(\mathbb{N};\times,m)$ is NP-hard for every integer $m > 1$.
\end{theorem}
\begin{proof}
If $m=k^n$ for some $k,n > 1$, then we can pp-define the constant relation $\{k\}$
since $x=k \Leftrightarrow \prod_{i=1}^k x=m$. Hence, we assume without loss of generality that $m \neq k^n$ for all $k,n > 1$.

We further know that $\mathrm{Div}_m$ is pp-definable in
$(\mathbb{N};\times,m)$, \mbox{i.e.} there is polynomial time reduction from $(\mathrm{Div}_m;\times,m)$ to $(\mathbb{N};\times,m)$. The core of $(\mathrm{Div}_m;\times,m)$ is some $(S;\times,m)$, where $\{1,m\} \subseteq S \subseteq \mathrm{Div}_m$ and the result follows from Lemma~\ref{finitehard}.
\end{proof}

\begin{theorem}
\label{thm:base-case}
Let $U$ be any subset of $\mathbb{N} \setminus \{0,1\}$ so that every $x \in U$ has a degree-one factor. Then  CSP$(\mathbb{N};\times,U)$ is NP-hard.
\end{theorem}
\begin{proof}
From Lemma~\ref{lem:2el-core}, for each $x \in U$, the core of $(\mathbb{N};\times,x)$ is the same (up to isomorphism). Fix some $m \in U$. We claim there is  a polynomial time reduction from CSP$(\mathrm{Div}_m;\times,m)$ to CSP$(\mathbb{N};\times,U)$, whereupon the result follows from Theorem~\ref{infinitemul}. 

To see the claim, take an instance $\phi$ of CSP$(\mathrm{Div}_m;\times,m)$ and build an instance $\psi$ of CSP$(\mathbb{N};\times,U)$ by adding an additional variable $v_m$, now substituting instances of $m$ for $v_m$, and adding the constraint $U(v_m)$. Correctness of the reduction is easy to see and the result follows.
\end{proof}

For $x \in \mathbb{N}\setminus\{0,1\}$, define its \emph{minimal exponent}, $\minexp(x)$, to be the smallest $j$ such that $x$ has a factor of $p^j$, for some prime $p$, but not a factor of $p^{j+1}$. Thus an integer with a degree-one factor has minimal exponent $1$. Call $x \in \mathbb{N}\setminus\{0,1\}$ \emph{square-free} if it omits all repeated prime factors. For a set  $U \subseteq  \mathbb{N}\setminus\{0,1\}$, define its \emph{basis}, $\mathrm{basis}(U)$ to be the set $\{\minexp(x) : x \in U\}$.
\begin{lemma}
\label{lem:atelier}
Let  $U \subseteq  \mathbb{N}\setminus\{0,1\}$, so that $\mathrm{basis}(U)$ is finite and $\mathrm{basis}(U)\neq \{1\}$. There is some set $X$ pp-definable in $(\mathbb{N};\times,U)$ so that $\mathrm{basis}(X)=\{1\}$. 
%is non-empty, does not contain $0$ and $\max(\mathrm{basis}(X))< \max(\mathrm{basis}(U))$.
\end{lemma}
\begin{proof}  
Let $r=\max(\mathrm{basis}(U))$ and take an element $x$ that witnesses this, of the form $q^r\cdot p_1^{a_1}\cdots p_k^{a_k}$, where $p_1,\ldots,p_k$ are prime and each is coprime to $q$ (which is square-free), and where $a_1,\ldots,a_k> r$. Set
\[\xi(y):=\exists z,x_1,\ldots,x_k \ U(z) \wedge y^r \cdot x_1^{a_1} \cdots x_k^{a_k} = z.\] 
We claim that $\xi$ defines a set of integers $X$ so that $\mathrm{basis}(X)$ has the desired property. The non-emptiness is clear since $1 \in \mathrm{basis}(X)$ by construction.

%Firstly, we will argue that $0 \notin \mathrm{basis}(X)$. This follows by maximality of $r$, taking some $s^{t}\cdot p_1^{b_1}\cdots p_\ell^{b_\ell}=z$ ($t\leq r$), since  some factor of $s$ will appear too many times or not at all.

%It is clear that  $\max(\mathrm{basis}(X))\leq \max(\mathrm{basis}(U))$ so we will now argue that $r \notin \mathrm{basis}(X))$. Suppose some element $w$ witnesses this. But then $(w^r)^r \cdot x_1^{a_1} \cdots x_k^{a_k}$ is overloaded with prime factors in $U$ since $r^2,a_1,\ldots,a_k>r$.

Firstly, we will by contradiction argue that $0 \notin \mathrm{basis}(X)$. Assume $0 \in \mathrm{basis}(X)$. This implies that $1 \in X$. Hence, $\exists z,x_1,...,x_k \ U(z) \wedge 1^r \cdot x_1^{a_1} \cdots x_k^{a_k} = z$, so $x_1^{a_1} \cdots x_k^{a_k} \in U$. It follows that $d=\min\{a_1,\ldots,a_k\} \in \mathrm{basis}(U)$. Now, $a_1,\ldots,a_k > r$ so $d > r$. This contradicts the fact that $r=\max(\mathrm{basis}(U))$.

We will now argue by contradiction that $1<s \notin \mathrm{basis}(X)$. Assume $s \in \mathrm{basis}(X)$. Then there exists a $t=q^s\cdot p_1^{a_1}\cdots p_k^{a_k} \in X$ where $s < a_1,\ldots,a_k$. Since $t \in X$, we know that $\exists z,x_1,\ldots,x_k \ U(z) \wedge t^s \cdot x_1^{a_1}\cdots x_k^{a_k}  = z$. Let $e= t^r \cdot x_1^{a_1}\cdots x_k^{a_k}  \in U$ as above. Let's expand $t$: $e = (q^s \cdot p_1^{a_1}\cdots p_k^{a_k})^r \cdot x_1^{a_1}\cdots x_k^{a_k}$. We see  that $\minexp(e)>r$ which contradicts the choice of $r$.
\end{proof}
\begin{example}
We provide an example of the construction of the previous lemma \emph{in vivo}. Let $U:=\{p^2, p^2q^3, p^4q^4r^8 : \mbox{ $p,q,r$ primes}\}$, so that  $\mathrm{basis}(U)=\{2,4\}$. Then $\xi(y):=\exists z,x \ U(z) \wedge y^4 \cdot x^8 = z$. We can now deduce $X:=\{p, pq^2: \mbox{ $p,q$ primes}\}$ and $\mathrm{basis}(X)=\{1\}$.
\end{example}

\begin{theorem}
\label{thm:main-hard}
Let $U \subseteq  \mathbb{N}\setminus\{0,1\}$ be so that $\mathrm{basis}(U)$ is finite. Then CSP$(\mathbb{N};\times,U)$ is NP-complete.
\end{theorem}
\begin{proof}
Membership of NP follows from Proposition~\ref{prop:in-NP}. We use the construction of the previous lemma to pp-define $X$ with $\mathrm{basis}(X)=\{1\}$. This allows us to polynomially reduce CSP$(\mathbb{N};\times,X)$ to CSP$(\mathbb{N};\times,U)$ by local substitution. NP-hardness for the former comes from Theorem~\ref{thm:base-case} and the result follows.
\end{proof}

\section{Final remarks}

In this paper we have provided a solution to the major open question from \cite{GlasserEtAl2010} as well as begun the investigation of CSPs associated with Skolem Arithmetic. However, the thrust of our work must be considered exploratory and there are two major directions in which more work is necessary.

A perfunctory glance at the results of Section~\ref{sec:circuits} shows that our bounds are not tight, and it would be great to see some natural CSPs in this region manifesting complexities such as Pspce-complete. It is informative to compare our Table~\ref{tablecsp} with Table 1 in \cite{GlasserEtAl2010}. Our weird formulation of these CSPs belies the fact there are more natural versions where, for $\mathcal{O} \subseteq \{-,\cap,\cup,+,\times\}$, we ask about CSP$(\mathcal{P}(\mathbb{N});\mathcal{O})$, where $\mathcal{P}(\mathbb{N})$ is the power set of $\mathbb{N}$, rather than the somewhat esoteric CSP$(\{\mathbb{N}\};\mathcal{O})$. Indeed, if we replace complement ``$-$'' by set difference ``$\setminus$'', these questions could also be phrased for just the finite sets of  $\mathcal{P}(\mathbb{N})$. These are all reasonable questions where we are not so sure of the boundary of decidability. Indeed, the complexity might be higher when the domain is  $\mathcal{P}(\mathbb{N})$, in comparison to $\{\mathbb{N}\}$. 

Meanwhile, the results of Section~\ref{sec:reducts-of-Skolem} need to be extended to a classification of complexity for all CSP$(\Gamma)$, where $\Gamma$ is a reduct of Skolem Arithmetic $(\mathbb{N};\times)$. We anticipate the first stage is to complete the classification for CSP$(\mathbb{N};\times,U)$ where $U$ is fo-definable in $(\mathbb{N};\times)$. We conjecture that Theorem~\ref{thm:main-hard} is tight in the sense that the outstanding cases are in P (of course this is trivial if $0$ or $1 \in U$). 
\begin{conjecture}
\label{conj:1}
Let $U \subseteq  \mathbb{N}\setminus\{0,1\}$ be fo-definable in $(\mathbb{N};\times)$, such that $\mathrm{basis}(U)$ is not finite. Then CSP$(\mathbb{N};\times,U)$ is in P.
\end{conjecture}
If this can be proved, the task is to extend to all problems of the form CSP($\Gamma$) where $\Gamma$ is an fo-expansion of $(\mathbb{N};\times)$. Since adding disequality already results in an NP-hard problem, we imagine this improvement is achievable. Finally, the task is to consider reducts, and not just fo-expansions, of $(\mathbb{N};\times)$. This might be harder, and will surely involve a specialised theorem of the form `Petrus' in \cite{dCSPs2}.

The mechanism through which we derived a polynomial algorithm to solve Proposition~\ref{prop:P} belies a close relationship between Skolem Arithmetic and Presburger Arithmetic (integers with addition). This was well-known to Mostowski \cite{Mostowski52}, who was able to prove decidability of Skolem Arithmetic through seeing it as a certain weak direct product of Presburger Arithmetic. Thus, it makes sense to study  CSP$(\Gamma)$, where $\Gamma$ is a reduct of Presburger Arithmetic $(\mathbb{N};+)$, in tandem with the program for Skolem Arithmetic. A first step would be to classify CSP$(\mathbb{N};+,U)$ where $U$ is fo-definable in $(\mathbb{N};+)$ and a group in Dresden \cite{DresdenNew} is working on this problem. Their polynomial algorithms may yet yield the solution to Conjecture~\ref{conj:1}.

%The functional CSP set-up was essential for the results of Section~\ref{sec:circuits}, whereas it was cosmetic for Section~\ref{sec:reducts-of-Skolem}.

%\bibliographystyle{acm}
%\bibliography{local}

\pagebreak

\section*{Appendix}

\subsection*{Missing proofs from Section~\ref{sec:circuits}}

We start with the observation that the equivalence of arithmetic formulas
reduces to functional CSPs. This yields several lower bounds for the CSPs.

\noindent \textbf{Proposition~\ref{prop_5147}}.
    For $\O \subseteq \{\com,\cup,\cap,+,\times\}$ it holds that
    $\EF{\O} \redlogm \CSP{\O}$.

\begin{proof}
    An $\EF{\O}$-instance $(F_1,F_2)$ is mapped
    to the $\CSP{\O}$-instance $F_1=F_2$.
\end{proof}

\

\noindent \textbf{Corollary~\ref{coro_5147}}. \
    \begin{enumerate}
        \item $\CSP{\{\com,\cup,\cap,+\}}$ and $\CSP{\{\com,\cup,\cap,\times\}}$
        are $\redlogm$-hard for $\cPSPACE$.
        \item $\CSP{\{\cup,\cap,+\}}$, $\CSP{\{\cup,\cap,\times\}}$,
        $\CSP{\{\cup,+\}}$, and \\ $\CSP{\{\cup,\times\}}$
        are $\redlogm$-hard for $\cPi{2}$.
    \end{enumerate}

\begin{proof}
    The statements follow from Proposition~\ref{prop_5147}
    and the following facts \cite{ghrtw10}:
    $\EF{\{\com,\cup,\cap,+\}}$ and $\EF{\{\com,\cup,\cap,\times\}}$
    are $\redlogm$-complete for $\cPSPACE$.

\noindent   $\EF{\{\cup,\cap,+\}}$, $\EF{\{\cup,\cap,\times\}}$,
    $\EF{\{\cup,+\}}$, and $\EF{\{\cup,\times\}}$
    are $\redlogm$-complete for $\cPi{2}$.
\end{proof}

\

%%%%%%%%%%%%%%%%%%%%%%%%%%%%%%%%%%%%%%%%%%%%%%%%%%%%%%%%%%%%%%%%%%%%%%%%%%%

\noindent \textbf{Proposition~\ref{prop_73478}}.
    $\CSP{+,\times}$ is $\redlogm$-hard for $\Sigma_1$.

\begin{proof}
    By the Matiyasevich-Robinson-Davis-Putnam theorem \cite{mat70,dpr61},
    there exists an $n \in \N$ and a multivariate polynomial $p$
    with integer coefficients such that for every $A \in \Sigma_1$
    there exists an $a \in \N$ such that
    $$x \in A \;\;\Longleftrightarrow\;\; \exists y \in \N^n, p(a,x,y)=0.$$
    In the equation $p(a,x,y)=0$ we can move negative monoms and
    negative constants to the right-hand side.
    This yields multivariate polynomials $l$ and $r$ with coefficients from $\N$
    such that 
    $$x \in A \;\;\Longleftrightarrow\;\; \exists y \in \N^n, l(a,x,y)=r(a,x,y).$$
    The right-hand side is a $\CSP{+,\times}$-instance.
    Hence $A \redlogm \CSP{+,\times}$ for every $A \in \Sigma_1$.
\end{proof}

\

\noindent \textbf{Proposition~\ref{prop_83651}}.
    $\CSP{\cup,\cap,+,\times} \in \Sigma_1$.

\begin{proof}
    It is decidable whether a given assignment satisfies
    a $\CSP{\cup,\cap,+,\times}$-instance.
    Hence testing the existence of a satisfying assignment is in $\Sigma_1$.
\end{proof}

\

\noindent \textbf{Proposition~\ref{prop_67126}}.
    $\CSP{\com,\cup,\cap,+,\times} \in \Sigma_2$.

\begin{proof}
    By Gla{\ss}er et al.\ \cite{ghrtw10},
    $\EC{\com,\cup,\cap,+,\times} \in \Delta_2$.
    Consider an arbitrary $\CSP{\com,\cup,\cap,+,\times}$-instance
    $x := \exists y \in \N^n [t_0=t_1 \wedge \cdots \wedge (t_{2m}=t_{2m+1}]$.
    It holds that
    $$x \in \CSP{\com,\cup,\cap,+,\times} \;\;\Longleftrightarrow\;\;
    \exists y \in \N^n [\EC{t_0,t_1} \wedge \cdots \wedge \EC{t_{2m},t_{2m+1}}].$$
    The right-hand side is a $\Sigma_2$ predicate.
\end{proof}

\

%%%%%%%%%%%%%%%%%%%%%%%%%%%%%%%%%%%%%%%%%%%%%%%%%%%%%%%%%%%%%%%%%%%%%%%%%%%

%%%%%%%%%%%%%%%%%%%%%%%%%%%%%%%%%%%%%%%%%%%%%%%%%%%%%%%%%%%%%%%%%%%%%%%%%%%

\noindent \textbf{Proposition~\ref{prop_6286}}.
    $\CSP{\times}$ is $\redlogm$-hard for $\cNP$.

\begin{proof}
    It is known that $\tn{3SAT} \redlogm \SC{\{\cap,\times\}}$ \cite{grtw10}.
    The reduction has the additional property that it outputs pairs $(C,b)$
    where the circuit $C$ is connected in the sense that from each gate there
    exists a path to the output gate.
    Hence it suffices to construct a $\redlogm$-reduction that
    works on $\SC{\{\cap,\times\}}$-instances $(C,b)$ where $C$ is connected.
    
    For such a pair $(C,b)$ we construct a $\CSP{\times}$-instance
    where each gate $g$ is represented by the variable $g$.
    Moreover, each gate $g$ causes the following constraints:
    If $g$ is an assigned input gate with value $k \in \N$,
    then we add the constraint $g=k$.
    For unassigned input gates no additional constraints are needed.
    If $g$ is a $\times$-gate with predecessors $g_1$ and $g_2$,
    then we add the constraint $g = g_1 \cdot g_2$.
    If $g$ is a $\cap$-gate with predecessors $g_1$ and $g_2$,
    then we add the constraints $g = g_1$ and $g = g_2$.
    If $g$ is the output gate, then this causes the additional constraint $g=b$.
    Finally, if $g_1, \ldots, g_n$ are the gates in $C$
    and $c_1, \ldots, c_m$ are the constraints described above,
    then the reduction outputs the $\CSP{\times}$-instance
    $\varphi := \exists g_1, \ldots, g_n [c_1 \wedge \cdots \wedge c_m]$.

    It remains to argue that for connected $C$ it holds that
    $$(C,b) \in \SC{\{\cap,\times\}} \;\Longleftrightarrow\; \varphi \in \CSP{\times}.$$

    Assume $(C,b) \in \SC{\{\cap,\times\}}$ and
    consider an assignment that produces $\{b\}$ at the output gate.
    Since $C$ is connected, each gate $g_i$ computes a singleton $\{a_i\}$.
    Hence $a_1, \ldots, a_n$ is a satisfying assignment for $\varphi$,
    which shows $\varphi \in \CSP{\times}$.

    Assume $\varphi \in \CSP{\times}$.
    Let $a_1, \ldots, a_n$ be a satisfying assignment for $\varphi$ and let
    $l$ be the number of $C$'s input gates.
    The constraints in $\varphi$ make sure that $C(a_1, \ldots, a_l)$ produces
    $\{a_i\}$ at gate $g_i$.
    In particular, $C(a_1, \ldots, a_l)$ produces $\{b\}$ at the output gate,
    which shows $(C,b) \in \SC{\{\cap,\times\}}$.
\end{proof}

\

\noindent \textbf{Proposition~\ref{prop_6287}}.
    $\CSP{+}$ is $\redlogm$-hard for $\cNP$.

\begin{proof}
    It suffices to show $\SC{\{+\}} \redlogm \CSP{+}$ \cite{grtw10}.
    The proof is similar to the proof of Proposition~\ref{prop_6286},
    but easier, since we have no $\cap$-gates and hence we do not need the
    assumption that $C$ is connected.
\end{proof}

\
%%%%%%%%%%%%%%%%%%%%%%%%%%%%%%%%%%%%%%%%%%%%%%%%%%%%%%%%%%%%%%%%%%%%%%%%%%%

\noindent \textbf{Proposition~\ref{prop_12481}}.
    $\CSP{\com,\cap,\cup}$ is $\redlogm$-complete for $\cNP$.

\begin{proof}
    Consider a $\CSP{\com,\cap,\cup}$-instance
    $\varphi := \exists x_1, \ldots, x_n
    [t_1 = t'_1 \wedge \cdots \wedge t_m = t'_m]$.
    We show that if $\varphi \in \CSP{\com,\cap,\cup}$, then it
    has a satisfying assignment $a = (a_1, \ldots, a_n)$ such that
    $a_1, \ldots, a_n \in \{0, \ldots, n-1\}$.
    This implies $\CSP{\com,\cap,\cup} \in \cNP$.

    Assume $\varphi \in \CSP{\com,\cap,\cup}$ and
    choose a satisfying assignment $a = (a_1, \ldots, a_n)$ such that
    $a' = \max \{a_1, \ldots, a_n\}$ is minimal.
    Assume that $a' \ge n$, we will show a contradiction.
    Let $b' = \min (\N - \{a_1, \ldots, a_n\})$ and note that $b'<n$.
    Let $b$ be the assignment that is obtained from $a$
    if all occurrences of $a'$ are replaced with $b'$.
    For any term $t$ in $\varphi$, the sets computed by $t$
    under the assignments $a$ and $b$ are denoted by $t_a$ and $t_b$, respectively.
    Observe that for all terms $t$ in $\varphi$ and
    all $x \in \N - \{a',b'\}$ it holds that:
    \begin{eqnarray*}
        x \in t_a &\Longleftrightarrow& x \in t_b\\
        a' \in t_a &\Longleftrightarrow& b' \in t_b\\
        b' \in t_a &\Longleftrightarrow& a' \in t_b
    \end{eqnarray*}
    It follows that for all atoms $t=t'$ in $\varphi$ it holds that
    $$t_a=t'_a \Longleftrightarrow t_b=t'_b.$$
    Therefore, $b$ is a satisfying assignment that is smaller than $a$,
    which contradicts the minimal choice of $a$.

    For the $\cNP$-hardness it suffices to show
    $\tn{3SAT} \redlogm \CSP{\com,\cap,\cup}$.
    On input of a $\tn{3CNF}$-formula $t=t(x_1, \ldots, x_n)$
    the reduction outputs the instance of $\CSP{\com,\cap,\cup}$
    $$\varphi := \exists x_1, \ldots, x_n [t' \cap \{1\} = \{1\}],$$
    where $t'$ is obtained from $t$ by replacing $\neg, \wedge, \vee$
    with $\com, \cap, \cup$, respectively.
    Every satisfying assignment for $t$ also satisfies $\varphi$.
    Conversely, if $a_1, \ldots, a_n$ is a satisfying assignment for $\varphi$,
    then we obtain a satisfying assignment for $t$
    if values greater than $1$ are replaced with $0$.
\end{proof}

\

\noindent \textbf{Proposition~\ref{prop_981281}}. \
    \begin{enumerate}
        \item $\CSP{\cap,+} \redmNP \CSP{+,=,\neq}$.
        \item $\CSP{\cap,\times} \redmNP \CSP{\times,=,\neq}$.
    \end{enumerate}

\begin{proof}
    We show the first statement, the proof of the second one is analogous.

    For a term $t$, let $t'$ be the term obtained from $t$
    if every subterm of the form $s_1 \cap s_2$ is replaced with $s_1$.

    We describe the $\redmNP$-reduction on input of a $\CSP{\cap,+}$-instance
    $$\varphi := \exists x_1, \ldots, x_n
    [t_0 = t_1 \wedge \cdots \wedge t_{2m} = t_{2m+1}].$$
    For each atom $t_{2i} = t_{2i+1}$,
    we guess nondeterministically whether $t_{2i} = t_{2i+1} \in \{\N\}$ or
    $t_{2i} = t_{2i+1} = \emptyset$.
    If we guessed $t_{2i} = t_{2i+1} \in \{\N\}$,
    then replace $t_{2i}$ with $t'_{2i}$, replace $t_{2i+1}$ with $t'_{2i+1}$,
    and for every subterm $s_1 \cap s_2$ that appears in $t_{2i}$ or $t_{2i+1}$
    add the constraint $s'_1 = s'_2$.
    If we guessed $t_{2i} = t_{2i+1} = \emptyset$,
    then guess a subterm $u_1 \cap u_2$ in $t_{2i}$,
    guess a subterm $u_3 \cap u_4$ in $t_{2i+1}$,
    remove the atom $t_{2i} = t_{2i+1}$, and
    add the constraints $u'_1 \neq u'_2$ and $u'_3 \neq u'_4$.
    The obtained formula $\psi$ is the result of the $\redmNP$-reduction.

    We argue that the described $\redmNP$-reduction reduces the $\CSP{\cap,+}$
    to the $\CSP{+,=,\neq}$.

    Assume $\varphi \in \CSP{\cap,+}$ and fix some satisfying assignment
    $a = (a_1, \ldots, $ $a_n) \in \N^n$.
    Consider the nondeterministic path of the reduction
    that for all atoms correctly guesses whether
    $t_{2i} = t_{2i+1} \in \{\N\}$ or $t_{2i} = t_{2i+1} = \emptyset$,
    and that for all $t_{2i} = t_{2i+1} = \emptyset$ guesses
    subterms $u_1 \cap u_2$ in $t_{2i}$ and $u_3 \cap u_4$ in $t_{2i+1}$
    such that $u_1, u_2, u_3, u_4 \in \{\N\}$, $u_1 \neq u_2$, and $u_3 \neq u_4$.
    If $t_{2i} = t_{2i+1} \in \{\N\}$, then $t'_{2i}=t_{2i}=t_{2i+1}=t'_{2i+1}$
    and hence the formula is still satisfied after replacing
    $t_{2i}$ with $t'_{2i}$ and $t_{2i+1}$ with $t'_{2i+1}$.
    Moreover, the added constraints $s'_1 = s'_2$ are satisfied,
    since in $t_{2i}$ and $t_{2i+1}$ all subterms $s_1 \cap s_2$ must be nonempty.
    If $t_{2i} = t_{2i+1} = \emptyset$,
    then after removing the atom $t_{2i} = t_{2i+1}$ and after
    adding the constraints $u'_1 \neq u'_2$ and $u'_3 \neq u'_4$
    the formula is still satisfied.
    So at the described nondeterministic path
    the reduction outputs a formula $\psi \in \CSP{+,=,\neq}$.

    Assume there is a nondeterministic path
    where the reduction outputs a formula $\psi \in \CSP{+,=,\neq}$.
    Consider a satisfying assignment $a$ for $\psi$,
    we claim that $a$ satisfies $\varphi$.
    If this is not true, then $\varphi$ must have an atom
    $t_{2i} = t_{2i+1}$ that is not satisfied by $a$.

    Case 1: At the path that produced $\psi$ we guessed that
    $t_{2i} = t_{2i+1} \in \{\N\}$.
    In this case we added the constraints $s'_1 = s'_2$,
    which ensure that $t_{2i}=t'_{2i}$ and $t_{2i+1}=t'_{2i+1}$.
    Hence under the assignment $a$ it holds that $t_{2i}=t'_{2i}=t'_{2i+1}=t_{2i+1}$,
    which contradicts the assumption that $t_{2i} = t_{2i+1}$
    is not satisfied by $a$.

    Case 2: At the path that produced $\psi$ we guessed that
    $t_{2i} = t_{2i+1} = \emptyset$.
    Here we added the constraints $u'_1 \neq u'_2$ and $u'_3 \neq u'_4$,
    which are satisfied by $a$.
    Hence under the assignment $a$ we have $t_{2i} = t_{2i+1} = \emptyset$,
    which contradicts the assumption that $t_{2i} = t_{2i+1}$
    is not satisfied by $a$.

    It follows that $a$ satisfies $\varphi$ and hence $\varphi \in \CSP{\cap,+}$.
\end{proof}

\

\noindent \textbf{Corollary~\ref{coro_981281}}.
    $\CSP{\cap,+}, \CSP{\cap,\times} \in \cNP$.

\begin{proof}
    $\CSP{+,=,\neq}$-instances and $\CSP{\times,=,\neq}$-instances are formulas of
    existential Presburger arithmetic and existential Skolem arithmetic,
    which are both decidable in $\cNP$ \cite{Sca84,gra89}.
    Now the statement follows from Proposition~\ref{prop_981281}.
\end{proof}

%%%%%%%%%%%%%%%%%%%%%%%%%%%%%%%%%%%%%%%%%%%%%%%%%%%%%%%%%%%%%%%%%%%%%%%%%%%

\end{document}